\documentclass[10pt,english,a4paper]{article} 
\usepackage{amsmath}  
\usepackage{amssymb}  
\usepackage[english]{babel}
\usepackage{graphicx}
\usepackage{verbatim}
\usepackage{charter,fullpage}

\usepackage[boxruled, vlined]{algorithm2e}

\newcommand{\ds}{\displaystyle}

\newcommand{\enabled}{\operatorname{enabled}}
\newcommand{\disabled}{\operatorname{disabled}}

\newcommand{\equals}{\stackrel{\mathrm{def}}{=}}

\newtheorem{theorem}{Theorem} 
\newtheorem{lemma}{Lemma} 
\newtheorem{corollary}{Corollary} 
\newtheorem{definition}{\bf Definition} 
\newenvironment{proof}[1][Proof]{\begin{trivlist}
\item[\hskip \labelsep {\bfseries #1}]}{\end{trivlist}}
\newtheorem{validity test}{Validity Test}

\newcommand{\upms}{\operatorname{upms}}
\newcommand{\FTP}{\operatorname{FTP}}
\newcommand{\FJP}{\operatorname{FJP}}
\newcommand{\step}{\widehat{\operatorname{step}}}
\newcommand{\minstep}{\operatorname{step}}

\newcommand{\SUMMSO}{\operatorname{SUM-MSO}}
\newcommand{\AUMMSO}{\operatorname{AUM-MSO}}
\newcommand{\SMMSO}{\operatorname{SM-MSO}}
\newcommand{\AMMSO}{\operatorname{AM-MSO}}
\newcommand{\CPU}{\operatorname{CPU}}
\newcommand{\fin}{\hfill{\small $\blacksquare$}}     
\begin{document}

\title{Scheduling Multi-Mode Real-Time Systems \\ upon Uniform Multiprocessor Platforms}


\author{
	Patrick Meumeu Yomsi$^1$, Vincent Nelis$^2$ and Jo\"{e}l Goossens\\
	Universit\'{e} Libre de Bruxelles (U.L.B.)\\
         50 Avenue F. D. Roosevelt, C.P. 212\\ 
	1050 Brussels - Belgium \\
	\{patrick.meumeu.yomsi, vnelis, joel.goossens\}@ulb.ac.be \\ 
}

\maketitle
\thispagestyle{empty}

\footnotetext[1]{Postdoctoral researcher of the Belgian National Science Foundation (F.N.R.S.).}
\footnotetext[2]{Supported by the Belgian National Science Foundation (F.N.R.S.) under a F.R.I.A. grant.}
\addtocounter{footnote}{2}
                                                                                                         

\begin{abstract}
In this paper, we address the scheduling problem of multi-mode real-time systems upon {\em uniform} multiprocessor platforms. We propose two transition protocols, specified together with their schedulability test, and provide the reader with two distinct upper bounds for the length of the transient phases during mode transitions, respectively for the cases where jobs priorities are known and unknown beforehand.
\end{abstract}

\vspace{-0.3cm}
\section{Introduction}\label{Introduction}

Over the years, the sporadic constrained-deadline task model \cite{Baruah0} has proven remarkably useful for the modeling of recurring processes that occur in hard real-time computer application systems, where the failure to satisfy any constraint may have disastrous consequences. The problem of scheduling a single set of such tasks so that all the deadlines are met has been widely studied in the literature. However, many applications exhibit multiple behaviors issued from several operating modes (e.g., an initialization mode, an emergency mode, a fault recovery mode, etc.), where each mode is characterized by its own set of tasks. During the execution of such {\em multi-mode} hard real-time systems, switching from the current mode (called {\em old-mode}) to another mode (called {\em new-mode}) requires to substitute the current executing task set with the one of the {\em new-mode}. This substitution introduces a transient phase, where the tasks of the old-mode may be scheduled together with those of the new-mode, which could lead to an {\em overload} that can jeopardize the system schedulability. In a multi-mode real-time system, any Mode Change Request (MCR) divides the timeline into the alternance of two phases: ($i$) A steady phase before the MCR occurs, and ($ii$) a transient phase during the mode change.

In the presence of a MCR, a task $\tau_i$ can be classified according to its behavior during the transition. Thus, we distinguish between old-mode and new-mode tasks, see \cite{Pedro} for details. If $\tau_i$ belongs to the old-mode, then it may need completing its whole execution or it could be aborted at the occurrence of the MCR. The abortion is performed when it is feasible without loss of data consistency. In this paper, we will assume that every old-mode job {\em must} complete its execution when a MCR occurs which is actually the worst-case. If $\tau_i$ belongs to the new-mode, then it could either be a completely new task, that is, it does not belong to the old-mode but is active in the new one, or it could be active in both modes, but with different or exactly the same parameters in the new mode. In this latter  case, it is said to be {\em mode-independent}. Due to the difficulty to guaranty the schedulability of such tasks, we only consider systems without mode-independent tasks in this research. 

If a transition protocol allows the management of {\em mode-independent} tasks, then it is said to be {\em with periodicity}, otherwise it is said to be {\em without periodicity}. Moreover if it allows to schedule new-mode tasks only when all old-mode ones are completed, then it is said to be {\em synchronous}, otherwise it is said to be {\em asynchronous}. 

\paragraph{Related work.} Up to now, the scheduling of {\em multi-mode} hard real-time systems has been much studied, particularly upon {\em uniprocessor} platforms, where there is only one shared processor available upon which all the jobs must be executed \cite{Real, Andersson0}. Recently, extensive efforts have been performed towards extending the uniprocessor results to multiprocessor platforms, where there are several shared processors available upon which jobs may execute. Sounds results have been obtained in the particular case of {\em identical} multiprocessor platforms \cite{Lopez1, Nelis1}. 

\paragraph{This research.} In this paper, we study the scheduling of {\em multi-mode} hard real-time systems upon {\em uniform} multiprocessor platforms. We propose two protocols -- a synchronous and an asynchronous one -- for managing the mode transitions. Note that the results presented here also hold for identical multiprocessor platforms as they are a special case of uniform multiprocessor platforms, in which the computing capacities of all processors are equal. 

\paragraph{Paper organization.} The remainder of this paper is structured as follows. Section~\ref{Model} presents the platform and system model, as well as the scheduler and the mode transition specifications that are used throughout the paper. Section~\ref{Definitions and observations} provides the reader with some useful definitions and observations. Section~\ref{sec:protocols} introduces two protocols -- a synchronous and an asynchronous one -- for managing the mode transitions during the execution of a multi-mode hard real-time systems upon a uniform multiprocessor platform. Section~\ref{sec:tests} provides sufficient conditions under which a given system can be executed on a given platform without missing any deadline. Section~\ref{sec:proofs} elaborates these conditions for the specific cases where jobs priorities are known and unknown beforehand. Section~\ref{Experimental results} presents experimental results. Finally, Section~\ref{Conclusion and future work} concludes the paper and proposes future work.

\section{Model of computation}\label{Model}

\subsection{Multi-mode real-time specifications}

We consider a multi-mode real-time system to be a set of $x$ operating modes $M^1, M^2, \ldots, M^x$ such that the operating mode $M^k$ has to execute the task set $\tau^k \equals \{\tau^k_1, \tau^k_2, \cdots, \tau^k_{n_k}\}$ consisting of $n_k$ tasks by following the scheduler ${\cal S}^k$. Each task $\tau^k_i$ is modeled by a sporadic constrained-deadline task characterized by three parameters $(C^k_i, D^k_i, T^k_i)$ where $C^k_i$ is the Worst Case Execution Time (WCET), $D^k_i \le T^k_i$ is the relative deadline and $T^k_i$ is the minimum inter-arrival time between two consecutive releases of $\tau^k_i$. These parameters are given with the interpretation that, during the execution of mode $M^k$, task $\tau^k_i$ generates a certain number of successive jobs $\tau^k_{i,j}$ with execution requirement of at most $C^k_i$ each, arriving at time $a^k_{i,j}$ such that $a^k_{i,j+1} - a^k_{i,j} \ge T^k_i$ and that must completes within $[a^k_{i,j}, d^k_{i,j})$ where $d^k_{i,j} \equals a^k_{i,j} + D^k_i$. Job $\tau^k_{i,j}$ is said to be {\em active} if and only if $a^k_{i,j} \le t$ and is not completed yet. More precisely, an active task is said to be {\em running} at time $t$ if it is allocated to a processor and is being executed. Otherwise the active task is in the ready queue of the operating system and it is said to be {\em ready}. We denote by {\em active}$(\tau^k, t)$, {\em run}$(\tau^k, t)$ and {\em ready}$(\tau^k, t)$ the subsets of active, running and ready tasks of $\tau^k$ at time $t$, respectively. Except during the transition phases, we assume that the system always runs in only one mode and that all the tasks are independent, i.e., there is no communication, no precedence constraint and no shared resource (except for the processors) between tasks. 

At run-time, every task in the system must be {\em enabled} before it can generate jobs, otherwise it is {\em disabled}. When all the tasks in $\tau^k$ are enabled and all the tasks in other modes are {\em disabled}, the system is said to be running in mode $M^k$. As such disabling $\tau_i^k$ prevents future job arrivals from $\tau_i^k$. We denote the subsets of {\em enabled} and {\em disabled} tasks of $\tau^k$ at time $t$ by $\enabled(\tau^k, t)$ and $\disabled(\tau^k, t)$, respectively.

We denote a Mode Change Request (MCR) from a given mode, say $M^i$, to a new mode, say $M^j$, by MCR($j$) and we denote the invoking time of a MCR($j$) by $t_{\operatorname{MCR}(j)}$. At the occurrence of a MCR the active old-mode jobs are called {\em rem-jobs} and must complete their execution as it has been assumed in Section~\ref{Introduction}. Note that the results that we are presenting in this paper also hold when some rem-jobs can be aborted at $t_{\operatorname{MCR}(j)}$ since such jobs do not jeopardize the system schedulability. 

Because the rem-jobs may cause an overload if the tasks of $\tau^j$ are immediately enabled upon a MCR($j$), the transition protocols sometimes have to delay the enablement of new-mode tasks until it is safe to do so. Consequently, we denote by $\mathcal{D}^j_k(M^i)$ the {\em relative enablement deadline} of task $\tau_k^j$ during the transition from mode $M^i$ to mode $M^j$, with the interpretation that the transition protocol {\em must} ensure that $\tau_k^j$ is not enabled after time $t_{\operatorname{MCR}(j)} + \mathcal{D}^j_k(M^i)$. The system enters mode $M^j$ as soon as all the rem-jobs are completed and all the tasks of $\tau^j$ are enabled.

\subsection{Platform specifications}

We consider the scheduling of multi-mode hard real-time systems upon a uniform multiprocessor platform comprised of $m$ processors. We denote the $m$-processor uniform platform by $\pi = [s_1, s_2, \ldots, s_m]$ where each processor $\pi_i, \: i = 1, 2, \ldots, m$ is characterized by its own {\em speed} or {\em computing capacity} $s_i$, with the interpretation that a job that executes on a processor $\pi_i$ for $t$ time units completes ($s_i \cdot t$) units of execution. For reasons of clarity and readability, we assume without any loss of generality that $s_j \le s_{j+1} \:\: \forall j = 1, 2, \ldots, m-1$ and $s_1 > 0$.

\subsection{Scheduler specifications}
\label{sec:scheduler_specs}

We consider that every mode $M^k$ uses its own scheduler denoted by ${\cal S}^k$ which can be either {\em Fixed-Task-Priority ($\FTP$)} or {\em Fixed-Job-Priority ($\FJP$)}. $\FTP$ schedulers assign a priority to each task at system design-time (i.e., before the system execution) and then, at run-time, every released job inherits the priority of the task it belongs to. Conversely, $\FJP$ schedulers determine the priority of each job at run-time and different jobs of the same task may have different priorities\footnote{According to these interpretations, the $\FTP$ schedulers are a particular case of the $\FJP$ schedulers, in which the priorities of all jobs issued from the same task are all equal to a same value determined beforehand.}. However, for both $\FTP$ and $\FJP$ schedulers, the priority of each job may not change between its release time and its completion time. Additionally, these two types of schedulers are assumed to be \emph{work-conserving} according to the following definition.
 
\begin{definition}[Work-conserving schedulers]\label{wc_sched}
A scheduler is said to be work-conserving upon an $m$-processor uniform platform if and only if it satisfies the following conditions:
\begin{itemize}
\item A processor cannot be idle if there are active ready jobs.
\item If at some time instant there are fewer than $m$ active ready jobs (recall that $m$ denotes the number of processors in the uniform multiprocessor platform), then the active ready jobs are executed upon the fastest processors. That is, at any time instant $t$ if the $j$'th-slowest processor is idled by the scheduler, then the $k$'th-slowest processor ($\forall k < j$) is also idled at $t$. 
\item Higher priority jobs are executed upon faster processors.
\end{itemize}
\end{definition}

\begin{lemma}\label{lem:step}
Let $J$ be any set of synchronous jobs and $\pi$ be any uniform platform. Let ${\cal S}$ denote the schedule of $J$ produced on $\pi$ by any work-conserving $\FJP$ scheduler. If $\minstep_j$ denotes the smallest instant in ${\cal S}$ where at least $j$ processors are idle, then it holds $\forall j=2, \ldots, m$ that 
\begin{equation}\label{equ:step}
\minstep_j \geq \minstep_{j-1}
\end{equation}
\end{lemma}

\begin{proof}
According to Definition~\ref{wc_sched}, when a job completes on (or migrates from) any processor $\pi_j$ with $j \in \left[ 2, m \right]$, the job (if any) executing on processor $\pi_{j-1}$ migrates to processor $\pi_j$. This directly leads to $\minstep_j \geq \minstep_{j-1}$ (see Figure~\ref{fig:staircase1}). \fin
\end{proof}

\begin{figure}[!h]
\centering
\includegraphics[scale=0.29]{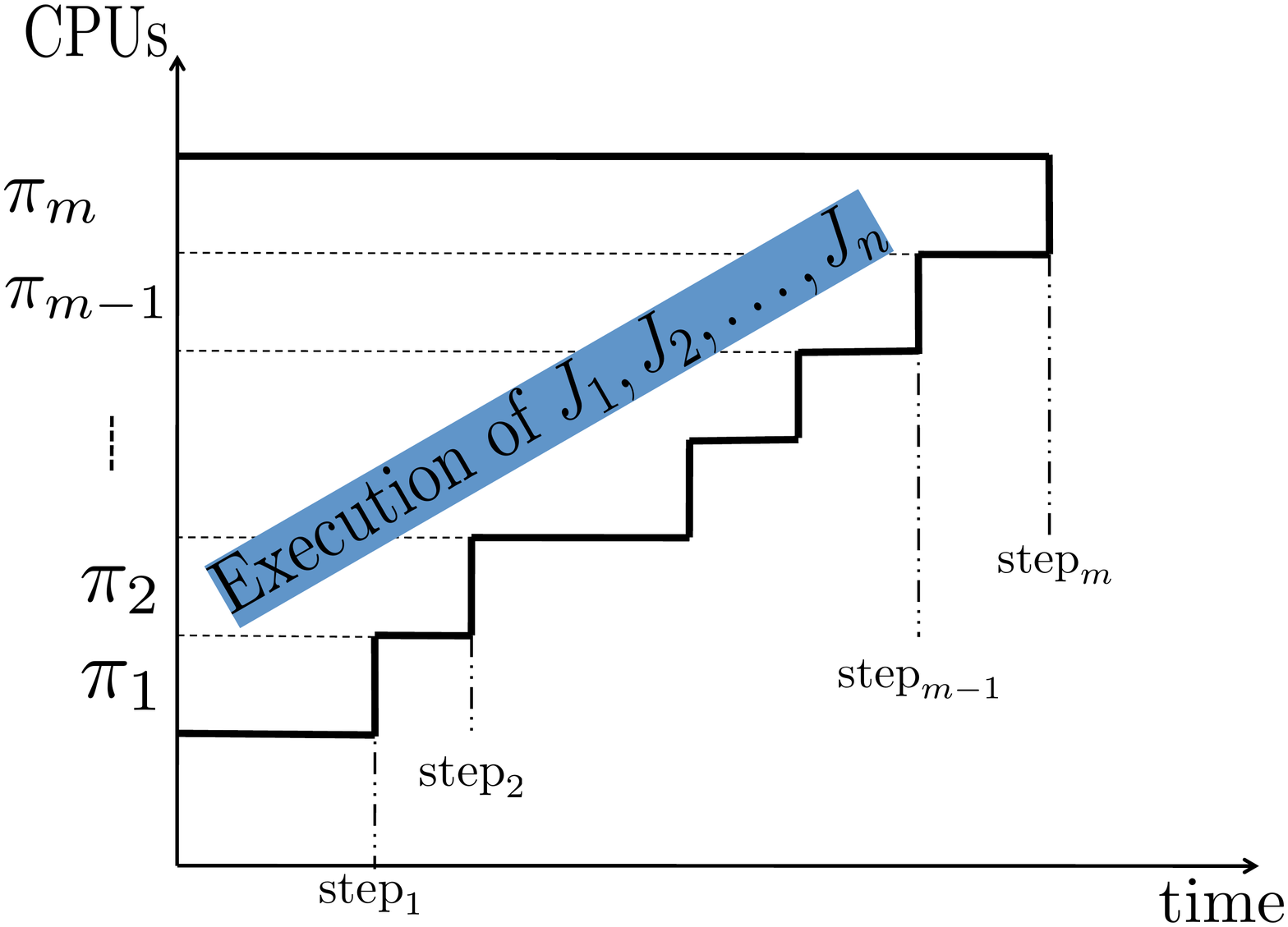}
\caption{Staircase defined by the $\minstep_j$}
\label{fig:staircase1}
\end{figure}

Hereafter, we assume without any loss of generality that every task set $\tau^k$ is schedulable in mode $M^k$ by the scheduler ${\cal S}^k$ on the $m$-processor uniform platform $\pi = [s_1, s_2, \ldots, s_m]$. This assumption allows us to only focus on the schedulability analysis of the system during the {\em transient phases} corresponding to mode transitions, rather than on the execution during the modes themselves. 

\subsection{Mode transition specifications}
\label{sec:modetrans_specs}

During the execution of a multi-mode real-time system in a given mode, say $M^i$, we consider that a Mode Change Request MCR($j$) to the new mode $M^j$ can be initiated by any task of $\tau^i$ or by the system itself, whenever it detects a change in the environment or in its internal state. At that time the system entrusts the scheduling decisions to the {\em transition protocol}. Such a protocol {\em immediately} disables all the old-mode tasks, which thus prevents the system of new jobs arrival from these tasks. The goal of the transition protocol is to complete all the rem-jobs and to enable all the new-mode tasks while meeting all the job and enablement deadlines. Once again we recall that we do not consider multi-mode real-time systems with mode-independent tasks. This study will be performed under the following assumptions during mode transitions: ($i$) {\em Job migration is permitted} (with no penalty). That is, a job that has been preempted on a particular processor may resume execution on the same or a different processor; ($ii$) {\em Job parallelism is forbidden}. That is, a job may execute on at most one processor at any given instant in time.

\section{Definitions and observations}
\label{Definitions and observations}

Before going any further in this paper, let us provide the reader with some useful definitions and observations. 

\begin{definition}[A valid protocol] 
A transition protocol is said to be {\em valid} for a given multi-mode real-time system if it always meets all the job and enablement deadlines during the transition from any mode of the system to any other one.  
\end{definition}

\begin{definition}[$S_{\pi}$]
Let $\pi = [s_1, s_2, \ldots, s_m]$ denote an $m$-processor uniform platform. We define $S_{\pi}$ as the sum of all the processor speeds, i.e., $S_{\pi} \equals \sum_{i=1}^{m} s_i$.
\end{definition}

\begin{definition}[Predictability]
Let $A$ denote a scheduler, and let $J = \{J_1, J_2, J_3, \ldots\}$ be a potentially infinite set of jobs, where each job $J_i = (a_i, c_i, d_i)$ is characterized by an arrival time $a_i$, a computation requirement $c_i$ and an absolute deadline $d_i$. Let $g_i$ (resp.\@\:$f_i$) denote the time at which job $J_i$ starts (resp.\@\:completes) its execution when $J$ is scheduled by $A$. Now consider any set $J'= \{J'_1, J'_2, J'_3, \ldots\}$ of jobs obtained from $J$ as follows. Job $J'_i$ is characterized by the arrival time $a_i$, a computation requirement $c'_i \le c_i$ and the absolute deadline $d_i$. Let $g'_i$ (resp.\@\:$f'_i$) denote the time at which job $J'_i$ starts (resp.\@\:completes) its execution when $J'$ is scheduled by $A$. Algorithm $A$ is said to be {\em predictable} if and only if for any set of jobs $J$ and any such $J'$ obtained from $J$, it is the case that $g'_i \le g_i$ and $f'_i \le f_i \:\: \forall i$. 
\end{definition}

\begin{lemma}[See \cite{CucuJo4}]\label{predict} 
Any work-conserving and FJP algorithm is predictable on uniform multiprocessor platforms. 
\end{lemma}

\begin{lemma}\label{rem-deadlines}
When a MCR($j$) occurs at time $t_{\operatorname{MCR}(j)}$ while the system is running in mode $M^i$, every rem-job issued from the tasks of $\tau^i$ meets its deadline when scheduled by $S^i$ upon an $m$-processor uniform platform.
\end{lemma}
\begin{proof}
From our assumptions, we know that the set of tasks $\tau^i$ of the mode $M^i$ is schedulable by $S^i$ upon an $m$-processor uniform platform, and from Lemma \ref{predict} we know that $S^i$ is predictable. When the MCR($j$) occurs at time $t_{\operatorname{MCR}(j)}$, all the tasks of $\tau^i$ are disabled. Disabling these tasks is equivalent to set the execution time of all their future jobs to zero, and since $S^i$ is predictable the deadline of every rem-job is still met. \fin
\end{proof}

\begin{lemma}\label{worst-case scenario}
At the occurrence of a MCR($j$) for the transition to an operating mode, say from mode $M^i$ to mode $M^j$, the {\em worst-case scenario} (in term of job completion time) is the one where all the rem-jobs issued from the tasks of $\tau^i$ are released simultaneously upon MCR($j$) with a computation requirement equals to their WCET each. 
\end{lemma}

\begin{proof}
The property is straightforward from Lemma~\ref{predict} and the fact that we only consider work-conserving schedulers in each operating mode. These ones are predictable. \fin
\end{proof}

\begin{definition}[Makespan]
Let $J = \{J_1, J_2, \ldots, J_n\}$ be a set of $n$ jobs, all released at time $t = 0$, with computation requirements $C_1, C_2, \ldots, C_n$, respectively. Let $\pi=~[s_1, s_2, \ldots, s_m]$ denote an $m$-processor uniform platform and $A$ be any scheduling algorithm. If $S$ denotes the schedule of $J$ produced by $A$ upon $\pi$ then the {\em makespan} is the earliest instant in $S$ at which all jobs of $J$ are completed.
\end{definition}

Very often, and especially when job priorities are unknown beforehand, determining the makespan of a set of jobs is a very challenging problem in scheduling theory. In the literature, extensive efforts have been made for the {\em minimum makespan scheduling problem} -- that is, to find a priority assignment for all the jobs of $J$ such that the makespan is minimized upon a given $m$-processor platform. Following the naming scheme introduced by Graham et al.\cite{GLLR1}, this problem is referred to as $P | | C_{\mbox{max}}$. However in this paper, we will focus on the {\em maximum makespan} that could be produced by a given set of $n$ jobs (all release at a same time), scheduled according to any work-conserving scheduler upon an $m$-processor uniform platform. 

For a given set of jobs, an intuitive idea for maximizing the makespan upon an $m$-processor uniform platform would be to execute, at any time, the longest ready job upon the slowest available processor, i.e., the shorter the computation requirement of a job, the higher its priority. However, we can show that this intuitive idea is {\em erroneous}, unfortunately\footnote{On the $2$-processor identical platform $\pi = [1, 1]$, the set of jobs $J = \{J_1, J_2, J_3, J_4\}$, all released at time $t = 0$ and such that $C_1= 2$, $C_2 = 3$, $C_3 = 5$ and $C_4 = 7$ provides a makespan of $10$ when $J_1 > J_2 > J_3 > J_4$, whereas the priority assignment $J_3 > J_1 > J_2 > J_4$ leads to a makespan of~$12$.}. An $\FJP$ assignment leading to the maximum makespan remains an open question.

Another intuitive idea would be to naively extend one of the results proposed in~\cite{Nelis1} for an $m$-processor {\em identical} platform, i.e.,
\begin{lemma}\label{upms} {(\bf Lemma $5$ in \cite{Nelis1})} 
Let $J = \{J_1, J_2, \ldots, J_n\}$ be a set of $n$ jobs, all released at time $t = 0$, with computation requirements $C_1, C_2, \ldots, C_n$ respectively, such that $C_1 \le C_2 \le \cdots \le C_n$. Suppose that these jobs are scheduled upon $m$ {\em identical} processors by a work-conserving scheduler $S$. Then, whatever the priority assignment of jobs, an upper bound on the makespan is given by
\begin{equation}
\label{equ:oldresult}
\footnotesize
\upms(J,m) \equals
  \left\{
          \begin{array}{ll}
	 C_n  & \quad \mathrm{if}\quad m \ge n \\ \\
          \ds\frac{1}{m} \ds\sum_{i=1}^{n}C_i + \left(1 - \ds\frac{1}{m}\right) \cdot C_n \quad & \quad \mathrm{otherwise}         	
          \end{array}
 \right.
\end{equation}
\end{lemma}
Naively extending Expression~\ref{equ:oldresult} leads to
\begin{equation}
\label{equ:naiveresult}
\footnotesize
\upms_0(J, \pi) \equals
  \left\{
          \begin{array}{ll}
	   C_n \slash s_{m-n+1}  & \:  \mathrm{if}\: m \ge n \\ \\
          \ds\frac{1}{S_\pi} \ds\sum_{i=1}^{n}C_i + \left(\ds\frac{1}{s_m} - \ds\frac{1}{S_\pi}\right) \cdot C_n  & \mathrm{otherwise}        
          \end{array}
 \right.
\end{equation}
and we can show that the intuitive idea used to derive this bound does not extend to uniform platforms, unfortunately\footnote{On the $3$-processor platform $\pi = [1, 2, 100]$, the set of jobs $J = \{J_1, J_2, J_3\}$, all released at time $t = 0$ and such that $C_1= 10, C_2 = 10$ and $C_3 = 100$ provides a maximum makespan of $1.1968$, reached when $J_1 > J_2 > J_3$. However, Expression~(\ref{equ:naiveresult}) provides $\upms_0(J, \pi) = 1.194175 <  1.1968$.}.

Now we are aware that neither the ``Shortest-Job-First'' policy nor $\upms_0(J, \pi)$ lead to the maximum makespan. In Section~\ref{sec:protocols}, we present the protocols $\SUMMSO$ and $\AUMMSO$ that are generalizations to uniform multiprocessor platforms of the protocols $\SMMSO$ and $\AMMSO$ respectively, defined for identical multiprocessor platforms \cite{Nelis1}. Then, we provide in Sections~\ref{sec:proofsFTP} and~\ref{sec:proofsFJP} two distinct {\em upper} bounds on the maximum makespan, for the cases where jobs priorities are known and unknown beforehand, i.e., $\FTP$ and $\FJP$ schedulers, respectively.

\section{Protocols $\SUMMSO$ and $\AUMMSO$}
\label{sec:protocols}

\paragraph{$\SUMMSO$.} The protocol $\SUMMSO$ which stands for ``Synchronous Uniform Multiprocessor Minimum Single Offset'' is an extension to {\em uniform} multiprocessor platforms of the protocol $\SMMSO$ defined for identical multiprocessor platforms \cite{Nelis1} to manage the rem-jobs during transition between any two operating modes. The main idea of $\SUMMSO$ is the following: upon a MCR($j$), every task of the current mode (say $M^i$) is disabled and the rem-jobs continue to be scheduled by $S^i$ upon the $m$ processors. \emph{When all of them are completed}, all the new mode tasks (i.e., the tasks of $\tau^j$) are simultaneously enabled. We refer the interested reader to~\cite{Nelis1} for a pseudo-code of this protocol.

\paragraph{$\AUMMSO$.} The protocol $\AUMMSO$ which stands for ``Asynchronous Uniform Multiprocessor Minimum Single Offset'' is an extension to uniform multiprocessor platforms of the protocol $\AMMSO$ defined for identical multiprocessor platforms \cite{Nelis1} to manage the rem-jobs during transition between any two operating modes. The main idea is the following: upon a MCR($j$), reduce the enablement delay applied to new-mode tasks by enabling them as soon as possible. Here, rem-jobs and new-mode tasks can be scheduled simultaneously during the transition phases according to the scheduler $S^{trans}$ with the following rules: ($i$) the priorities of the rem-jobs are assigned according to the old-mode scheduler; ($ii$) the priorities of the new-mode tasks are assigned according to the new-mode scheduler, and ($iii$) every rem-job always has a higher priority than every new-mode task.

Upon a MCR($j$), all the old-mode tasks, say of mode $M^i$, are disabled and the rem-jobs continue to be scheduled by $S^i$. Whenever the lowest priority rem-job migrates to a faster processor due to the completion of a higher priority one, (say at time $t$), some processors become available and thus the protocol $\AUMMSO$ immediately enables some new-mode tasks; contrary to the protocol $\SUMMSO$ which waits for the completion of {\em all} the rem-jobs. In order to select the new-mode tasks to enable at time $t$, $\AUMMSO$ uses the following heuristic: it considers every disabled new-mode task in increasing order of their enablement deadline and it enables those which can be scheduled by $S^j$ upon the current available CPUs (i.e., the CPUs which are not running a rem-job and which are therefore available for executing some new-mode tasks). 

Let $\pi = [s_1, s_2, \ldots, s_m]$ be an $m$-processor uniform platform with processor capacities such that $s_j \le s_{j+1}$ for all $j, \: 1 \le j \le m-1$. Let $S$ be any work-conserving $\FJP$ scheduler. Let $\tau^\ell$ be a set of tasks. 
We denote by $\pi^S(\tau^\ell)$ the subset of processors running a job of $\tau^\ell$ when $\tau^\ell$ is scheduled by $S$ upon platform $\pi$. 
We denote by $\CPU(\pi, S, \tau^\ell)$ the binary function defined by:
\[
\CPU(\pi, S, \tau^\ell) \equals
  \left\{
          \begin{array}{ll}
	   1 & \quad \mathrm{if} \:\: \tau^\ell \:\: \mbox{is schedulable by} \:\: S \:\: \mbox{upon} \:\: \pi \\
            0 & \quad \mathrm{otherwise}
         \end{array}
 \right.
 \]
This function is useful as we must always guarantee that all the deadlines are met for all the jobs in the system. To the best of our knowledge, there is no {\em efficient necessary} and {\em sufficient} schedulability test for most multiprocessor schedulers upon uniform platforms. However sufficient schedulability tests exist for scheduler such as EDF and DM \cite{BaJo08_1, BaJo08_2}. Algorithm \ref{AUM-MSO protocol} gives the pseudo-code of the $\AUMMSO$ protocol in a more formal way.

\begin{algorithm}
\footnotesize
\KwIn{$M^i$: old-mode; $M^j$: new-mode}
\KwOut{a safe release of the new-mode tasks}
\Begin{
	Assign priorities to jobs according to $S^{trans}$ \;
	Sort $\disabled(\tau^j, t)$ by increasing enablement deadlines \;
	$\pi^{\operatorname{old}} \leftarrow \pi$ \; 
	$\pi^{\operatorname{new}} \leftarrow \emptyset$ \; 
	
	\While{($\pi^{\operatorname{old}} \neq \emptyset$)}{
		At job (say $J_k$) completion time $t$, a subset (say $\pi^{\operatorname{avl}} = \{\pi_{\operatorname{s}}, \ldots, \pi_{\operatorname{f}}\}$) of ($f-s+1$) slowest processors may become available:\\
	
		\If{$(J_k \in \tau^i$ {\bf and} $ready(\tau^i, t) = \emptyset$)}{
			$\pi^{\operatorname{old}} \leftarrow \pi^{\operatorname{old}} \backslash \pi^{\operatorname{avl}}$ \;
		  	$\pi^{\operatorname{new}} \leftarrow \pi^{\operatorname{new}} \cup \pi^{\operatorname{avl}}$ \; 
			\ForAll{$\tau^j_r \in \disabled(\tau^j, t)$}{
				$\tau^{\operatorname{temp}} \leftarrow \enabled(\tau^j, t) \cup \{\tau^j_r\}$ \;
				\If{$(\CPU(\pi^{\operatorname{new}}, S^{j}, \tau^{\operatorname{temp}}) \neq 0)$}{
					enable $\tau^j_r$ \;
				}	
			}
		}
	}
	\lIf{($active(\tau^i, t) = \emptyset$)}{
		enter mode $M^j$ \;
	}	
}
\caption{$\AUMMSO$ protocol}
\label{AUM-MSO protocol}
\end{algorithm}

\section{Validity tests for protocols $\SUMMSO$ and $\AUMMSO$}
\label{sec:tests}

In the previous section we have defined the transition protocols $\SUMMSO$ and $\AUMMSO$. Now, we need to establish a {\em validity test} -- that is, a condition based on the tasks and platforms characteristics that indicates {\em a priori} whether the given system will always comply its expected requirement during every mode change. To do so, we proved in Lemma \ref{rem-deadlines} that disabling the old-mode tasks upon a MCR does not jeopardize the schedulability analysis of the rem-jobs, when they continue to be scheduled by using the old-mode scheduler specifications upon the $m$ processors. In Lemma \ref{worst-case scenario} we defined the worst-case scenario. 

\subsection{Validity test for $\SUMMSO$}

Thanks to Lemma \ref{rem-deadlines}, the deadline of every rem-job is always met while using $\SUMMSO$ during the transition phases. Thereby, $\SUMMSO$ is valid for a given multi-mode real-time system if, for every mode change, the maximal transition delay which could be produced by the rem-jobs is not larger than the minimal enablement deadline of the new-mode tasks. Thanks to Lemma \ref{worst-case scenario}, the transition delay which is actually equal to the completion time of all the rem-jobs is maximal when they are released simultaneously, with execution requirements equal to their WCETs. This leads to the following validity test.

\begin{validity test}
For any multi-mode real-time system $\tau$ and any uniform platform $\pi$, protocol $\SUMMSO$ is valid provided $\forall i,j$ with $i \neq j$:
\[ \upms(\tau^i, \pi) \le \min_{1 \le k \le n_j} \left(\mathcal{D}_k^j(M^i)\right) \]
\noindent where $\upms(\tau^i, \pi)$ is an upper-bound on the makespan and is defined in Sections~\ref{sec:proofsFTP} and~\ref{sec:proofsFJP} for both $\FTP$ and $\FJP$ schedulers, respectively. 
\end{validity test}

\subsection{Validity test for $\AUMMSO$}

The main idea to know whether $\AUMMSO$ is valid for a given system $\tau$ and platform $\pi$ is to simulate Algorithm~\ref{AUM-MSO protocol} for every possible mode transition, while considering the worst-case scenario for each one -- the scenario where the new-mode tasks are enabled as late as possible. From our definition of protocol $\AUMMSO$, we know that every instant at which some new-mode tasks are enabled corresponds to an instant at which a processor has no more rem-job to execute, i.e., the ``step instants'' $\minstep_j$ of the staircase depicted in Figure~\ref{fig:staircase1}. Consequently, the largest instants at which new-mode tasks could be enabled are the upper-bounds $\step_j$ on the instants $\minstep_j$ and can be determined by considering only the schedule of the rem-jobs. These upper-bounds are defined for both $\FTP$ and $\FJP$ schedulers in Sections~\ref{sec:proofsFTP} and~\ref{sec:proofsFJP}, respectively. Notice that it results from this notation that $\step_m = \upms(\tau^i, \pi)$ and the validity test for $\SUMMSO$ can be rewritten as $\step_m \le \min_{1 \le k \le n_j} (\mathcal{D}_k^j(M^i))$ $\forall i,j$ with $i \neq j$. 

The details for the validity of the transition protocol AUM-SMO are provided by Algorithm \ref{Validity test}. The correctness of our validity algorithm derives directly from the fact that every instant at which a new-mode task is enabled in Algorithm \ref{Validity test} is as large as possible. Notice that, since our validity test considers only the worst-case scenario for every mode transition, it could be some time-instants (in any mode transition) during the actual execution of the system at which the set of already enabled tasks benefits from a larger number of available processors than in the worst-case scenario. However, we prove in Lemma~\ref{no-anomalies} that it does not jeopardize the system schedulability.

\begin{lemma}\label{no-anomalies}
Any predictable and work-conserving scheduler that is able to schedule a task set $\tau$ upon a uniform platform $\pi = [s_1, \ldots, s_m]$ is also able to schedule $\tau$ upon any platform $\pi^{*}$ such that $\pi \subseteq \pi^{*}$.
\end{lemma}

\begin{proof}
The proof is made by contradiction. To do so, it is sufficient to show the Lemma for $\pi^{*} = [s_1, \ldots, s_m, s_{*}]$ where $s_{*} \ge s_m$. Suppose there exists a task set $\tau$ that is schedulable with a predictable and work-conserving scheduler ${\cal S}$ upon $\pi$, but not upon $\pi^{*} \supseteq \pi$. Consider the schedule of a particular instance ${\cal I}$ of $\tau$ upon $\pi^{*}$ leading to a deadline miss, and let ${\cal I^{*}}$ be another instance of $\tau$ derived from ${\cal I}$ reducing each job requirement by the amount of time each job executes upon the sub-platform $\pi^{*} \backslash \pi$, i.e., upon $\pi_{*}$. Since the scheduler is work-conserving, the schedule of ${\cal I}$ by ${\cal S}$ upon the processors in common with $\pi$ is the same as the one that would be produced by ${\cal S}$ for ${\cal I^{*}}$ upon platform $\pi$. Since a deadline is missed in the schedule of ${\cal I}$ upon $\pi^{*}$, then a deadline is missed also in the schedule of ${\cal I^{*}}$ upon $\pi$. But since the scheduler is predictable, a deadline would be missed on $\pi$ even with the more demanding instance ${\cal I}$, leading to a contradiction. The lemma follows. \fin
\end{proof}

\begin{algorithm}
\footnotesize
\KwIn{A multi-mode hard real-time system $\tau = \left\{ \tau^1, \tau^2, \ldots, \tau^x \right\}$}
\KwOut{a Validity Test of the transition protocol}
\Begin{
	\ForAll{$i, j \in [1, x]$ such as $i \neq j$}{
		$\tau^{\disabled} \leftarrow \tau^j$ \; 
		$\tau^{\enabled} \leftarrow \emptyset$ \;
		Sort $\tau^{\disabled}$ by increasing enablement deadlines \;
		
		\For{($k=1; k \le m; k++$)}{
			\ForAll{$\tau^j_r \in \tau^{\disabled}$}{
				\lIf{$\mathcal{D}_r^j(M^i) < \step_k$}{
					{\bf return} False \;
				}
				$\tau^{\operatorname{temp}} \leftarrow \enabled(\tau^j, t) \cup \{\tau^j_r\}$ \;
				$\pi^{\operatorname{new}} \leftarrow \pi \backslash \pi^{S^i}(\tau^i)$ \;
				\If{($\CPU(\pi^{\operatorname{new}}, S^{j}, \tau^{\operatorname{temp}}) \neq 0)$}{
					$\tau^{\enabled} \leftarrow \tau^{\enabled} \cup \{\tau_r\}$ \;
					$\tau^{\disabled} \leftarrow \tau^{\disabled} \backslash \{\tau_r\}$ \;
				}
			}
		}
	}
	{\bf return} True \;
}
\caption{Validity Test for $\AUMMSO$}
\label{Validity test}
\end{algorithm}

In the next sections, we determine the upper-bounds $\step_j$, $\forall j = 1, 2, \ldots, m$ for both $\FTP$ and $\FJP$ schedulers.

\section{Determination of the upper-bounds $\step_j$}
\label{sec:proofs}

\subsection{Upper-bounds for $\FTP$ schedulers}
\label{sec:proofsFTP}

We recall that for $\FTP$ schedulers job priorities are known beforehand. 

\begin{definition}[$t_{j}^{i}$]\label{def:tij}
The time-instant $t_{j}^{i}$ denotes the earliest instant in the schedule of the $i$ highest priority jobs $J_{1}, \ldots, J_{i}$ where at least $j$ processors are idle (the processors $\pi_{1}, \ldots, \pi_{j}$, since we consider work-conserving schedulers).
\end{definition}

In Theorem~\ref{theo3} we provide the exact values of $t_{j}^{i}$ ($\forall \: i \in \{1, 2, \ldots, n\}$ and $\forall \: j \in \{1, 2, \ldots, m\}$) when each job executes for its WCET. As such we provide the exact value of $t_{m}^{n}$ which corresponds to the exact makespan for the scheduling of $J = \{J_{1},\ldots,J_{n}\}$ upon the platform $\pi = [s_{1}, \ldots, s_{m}]$.

\begin{theorem}\label{theo3}
Let $J = \{J_1, J_2, \ldots, J_n\}$ be a set of $n$ jobs, all released at time $t = 0$, with computation requirements $C_1, C_2, \ldots, C_n$, respectively. Suppose that these jobs are scheduled according to a work-conserving $\FTP$ scheduler. Suppose that $J$ is ordered in decreasing-priority order then $t_{i}^{j}$ is inductively defined as follows:
 (initialization) $t_{j}^{0} = 0, \forall j \leq m$ and $t_{m+1}^{i} = \infty, \forall i$, (iteration) 
\begin{equation*}
t_{j}^{i} = 
\begin{cases}
t_{j}^{i-1} \text{\ \ \ \ \ \ \ \ \ \ \ \ if } t_{j}^{i-1} = t_{j+1}^{i-1} \quad \\ 
t_{j+1}^{i-1} \text{\ \ \ \ \ \ \ \ \ \ \ \ else if } C_{i} \geq \sum_{\ell=1}^{j} (t_{\ell+1}^{i-1} - t_{\ell}^{i-1})\cdot s_{\ell} \quad \\ 
t_{j}^{i-1} + \ds\frac{1}{s_j} \cdot \left(C_{i} - \ds\sum_{\ell=1}^{j-1} (t_{\ell+1}^{i-1} - t_{\ell}^{i-1})\cdot s_{\ell}\right) \quad \text{otherwise}
\end{cases}
\end{equation*}
\end{theorem}

\begin{proof} 
Initially, the $m$ processors are idle, consequently $t_{j}^{0} = 0, \forall j \leq m$. We find convenient to define $t_{m+1}^{i} \equals \infty, \forall i$, which means that we have at most $m$ processors available.

Now we will prove the correctness of the value of $t_{j}^{i}$ ($\forall j \leq m$) assuming that $t_{j}^{i-1}$ are defined ($\forall j \leq m+1$). The time-instants $t_{j}^{i-1}$ define a staircase as illustrated in Figure~\ref{fig:staircase} for the scheduling of jobs $J_{1}, \ldots, J_{i-1}$. As such, job $J_{i}$ can only progress during the grey areas. Consequently we have to distinguish between two cases:
\begin{enumerate}
\item $t_{j}^{i-1} = t_{j+1}^{i-1}$, i.e., at least one faster processor becomes available at time $t_{j}^{i-1}$ (the grey area on processor $\pi_{j}$ is void in that case), the job $J_{i}$ will be executed (if not completed) upon a faster processor, consequently the first time-instant where at least $j$ processors are idle remains unchanged and $t_{j}^{i} = t_{j}^{i-1}$.

\item Otherwise, $J_{i}$ will be scheduled upon processor $\pi_{j}$ while no faster processors become available or $J_{i}$ completes. Remark that $\sum_{\ell=1}^{j} (t_{\ell+1}^{i-1} - t_{\ell}^{i-1})\cdot s_{\ell}$ corresponds to the grey area on processors $\pi_{1}, \dots, \pi_{j}$. Hence the two subsequent sub-cases follow.
\end{enumerate}
\fin
\end{proof}

\begin{figure}[!h]
\centering
\includegraphics[scale=0.3]{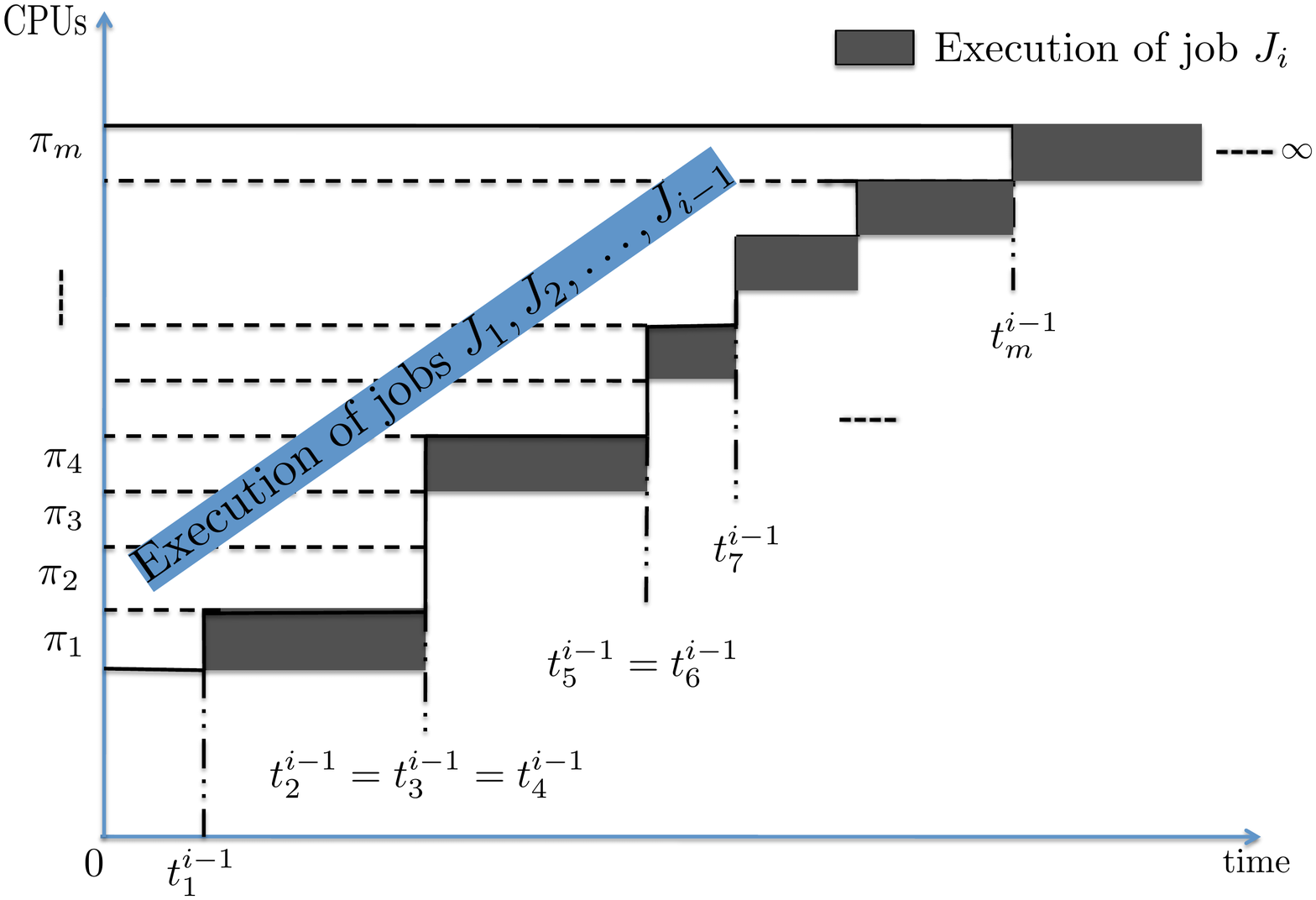}
\caption{Staircase defined by the $t_{j}^{i-1}$}
\label{fig:staircase}
\end{figure}

\begin{corollary}
Each instant $t^n_{j}$ ($\forall j = 1, 2, \ldots, m$) defined in Definition~\ref{def:tij} is an upper-bound on the ``step-instant'' $\minstep_j$ defined in Lemma~\ref{lem:step}, i.e., $\minstep_j \leq t^n_j$ ($\forall j = 1, 2, \ldots, m$). Therefore, the instants $t^n_{j}$ ($\forall j = 1, 2, \ldots, m$) computed in Theorem~\ref{theo3} can be used as the instants $\step_j$ in the validity test provided by Algorithm~\ref{Validity test}.
\end{corollary}

\begin{proof}
From Theorem~\ref{theo3}, the instants $t^n_{j}$ ($\forall j = 1, 2, \ldots, m$) are derived from the schedule in which every job executes for its WCET (we denote this ``worst-case'' schedule $S^{\operatorname{wc}}$ hereafter). Now, suppose by contradiction that in the actual schedule $S^{\operatorname{act}}$ (during the system execution), there exist $j \in \left[ 1, m \right]$ such that $\minstep_j$ in $S^{\operatorname{act}}$ is strictly larger than $t^n_j$ from $S^{\operatorname{wc}}$. This implies that, within the time interval $[ t^n_j, \minstep_j [$, there are \emph{at least} $(m-j+1)$ running jobs in $S^{\operatorname{act}}$ while there are \emph{at most} $(m-j)$ running jobs in $S^{\operatorname{wc}}$. Therefore, within $[ t^n_j, \minstep_j [$, at least one job (say $J_{\ell}$) is not completed yet in $S^{\operatorname{act}}$ whereas it is already completed in $S^{\operatorname{wc}}$. But since we consider only work-conserving schedulers and since in $S^{\operatorname{act}}$ the execution requirement of $J_{\ell}$ can only be lower than or equal to that in $S^{\operatorname{wc}}$, the fact that $J_{\ell}$ completes later in $S^{\operatorname{act}}$ than in $S^{\operatorname{wc}}$ leads to a contradiction with the predictability. \fin
\end{proof}



\subsection{Upper-bounds for $\FJP$ schedulers}
\label{sec:proofsFJP}



We recall that for $\FJP$ schedulers job priorities are unknown beforehand. We assume without any loss of generality that we always have $m \le n$ as the problem in the case where $m \ge n$ reduces to the same problem upon the $n$ fastest processors. In Lemma~\ref{earliest tidle}, we first determine a lower bound on the smallest instant in the schedule of $J$ where at least $j$ CPUs (with $j = 1, 2, \ldots, m$) are idle. Then in Theorem~\ref{latest tidle} we determine an upper bound on the maximum makespan that could be produced by $J$, this is given by $\step_m$.

\begin{lemma}\label{earliest tidle}
Let $J = \{J_1, J_2, \ldots, J_n\}$ be a set of $n$ jobs with computation requirements $C_1, \ldots, C_n$, respectively, such that $C_1 \le \cdots \le C_n$. A lower bound $\hat{L}_j$ on the smallest instant $\minstep_j$ at which at least $j$ CPUs (with $j = 1, 2, \ldots, m$) are idle is given by
\begin{equation}
\hat{L}_j \equals \ds\frac{1}{S_{\pi}} \cdot \ds\sum_{k=1}^{n-m+j} C_k \:\: \mbox{where} \:\: S_{\pi} \equals \ds\sum_{i=1}^{m} s_i
\end{equation}
\end{lemma}

\begin{proof}
For the schedule ${\cal S}$ obtained from a work-conserving $\FJP$ scheduler, let $\minstep_j$ denote the smallest instant in ${\cal S}$ at which at least $j$ processors are idle. According to Definition~\ref{wc_sched}, at most $(m-j)$ jobs are not completed at time $\minstep_j$, meaning that \emph{at least} $(n - m + j)$ are already completed. If $J'$ denotes any subset of $r \geq (n - m + j)$ jobs of $J$, then a lower bound $t$ on the instant at which the $r$ jobs of $J'$ are completed is given by $t \equals \ds\frac{1}{S_{\pi}} \cdot \ds\sum_{J_k \in J'} C_k$. Obviously, since $C_1 \leq C_2 \leq \ldots \leq C_n$, $t$ is minimal for $J' = \{ J_1, J_2, \ldots, J_{n - m + j}\}$, i.e., $t =  \hat{L}_j$. The lemma follows. \fin
\end{proof}

\begin{theorem}\label{latest tidle}
Let $J = \{J_1, J_2, \ldots, J_n\}$ be a set of $n$ jobs with computation requirements $C_1, \ldots, C_n$, respectively, such that $C_1 \le \cdots \le C_n$. An {\em upper bound} $\step_j$ on the smallest instant at which at least $j$ CPUs (with $j = 1, 2, \ldots, m$) are idle is given by
\begin{equation}
\step_j \equals \frac{\sum_{k=1}^{n} C_k - \sum_{k=1}^{j-1} \hat{L}_k \cdot s_k}{\sum_{k=j}^m s_k}
\end{equation}
\end{theorem}

\begin{proof}
Consider the following notations: ($i$) ${\cal S}$ denotes the schedule of the $n$ jobs obtained from a work-conserving $\FJP$ scheduler; ($ii$) $\minstep_j$ (with $j = 1, 2, \ldots, m$) denotes the smallest instant in ${\cal S}$ at which at least $j$ processors are idle, and ($iii$) $W_j$ denotes the amount of work executed on CPU $\pi_j$ within $[ 0, \minstep_j ]$, i.e., $W_j \equals \minstep_j \cdot s_j$. Let $k$ be any integer in $\left[ 1, m \right]$ and suppose by contradiction that $\minstep_k > \step_k$. By definition of the $W_j$, we know that 
\begin{equation}\label{equ:workissumci}
\sum_{j = 1}^m W_j = \sum_{j = 1}^n C_j
\end{equation}
Furthermore, we know that $\sum_{j = 1}^m W_j = \sum_{j = 1}^m \minstep_j \cdot s_j = \sum_{j = 1}^{k-1} (\minstep_j \cdot s_j) + \sum_{j = k}^{m} (\minstep_j \cdot s_j)$. Since we know from Lemma~\ref{lem:step} that we have $\minstep_{j+1} \geq \minstep_j$ $\forall~j=1, 2, \ldots, m-1$, it holds that
\begin{eqnarray}
\sum_{j = 1}^m W_j & \geq & \sum_{j = 1}^{k-1} (\minstep_j \cdot s_j) + \sum_{j = k}^{m} (\minstep_k \cdot s_j) \nonumber 
\end{eqnarray}
By assumption we have $\minstep_k > \step_k$, leading to
\begin{small}
\begin{eqnarray}
\sum_{j = 1}^m W_j & > & \sum_{j = 1}^{k-1} (\minstep_j \cdot s_j) + \step_k \cdot \sum_{j = k}^{m} s_j \nonumber \\
& > & \sum_{j = 1}^{k-1} (\minstep_j \cdot s_j) + \frac{\sum_{j=1}^{n} C_j - \sum_{j=1}^{k-1} \hat{L}_j \cdot s_j}{\sum_{j=k}^m s_j} \cdot \sum_{j = k}^{m} s_j \nonumber \\
& > & \sum_{j=1}^{n} C_j + \sum_{j = 1}^{k-1} \left((\minstep_j - \hat{L}_j ) \cdot s_j \right) \nonumber
\end{eqnarray}
\end{small}
Since from Lemma~\ref{earliest tidle} it holds that $\hat{L}_j \leq \minstep_j$ $\forall j=1, 2, \ldots, m$, it yields $\sum_{j = 1}^m W_j > \sum_{j=1}^{n} C_j$ leading to a contradiction with Equality~\ref{equ:workissumci}. The theorem follows. \fin
\end{proof}


For experimental purpose, let us recall the definition of the parameter $\lambda_\pi$ \cite{SGB2001} for an $m$-processor uniform platform $\pi = [s_1, s_2, \ldots, s_m]$: $\lambda_\pi \equals \max_{j=1}^{m} \left\{\frac{\sum_{k=1}^{j-1} s_k}{s_j} \right\}$. Note that parameter $\lambda_\pi$ measures the ``degree'' by which $\pi$ differs from an identical multiprocessor platform. 


\section{Experimental results}\label{Experimental results}

In this section, we report on the results of experiments conducted using the theoretical results presented in Section~\ref{sec:proofsFJP} (since the upper-bounds presented in Section~\ref{sec:proofsFTP} can be considered as exact if every job executes for its WCET).
The considered set of jobs $J$ is composed of $10$ jobs of \emph{undetermined} priority and 
the platform $\pi$ is composed of $4$ processors with computing capacities varying within $\left[1, 101 \right]$ with an increment of $10$. 

During the simulation, all possible combinations of the processors speeds are considered for the 4 CPUs. For every assignation of processors speed, we determine the corresponding parameter $\lambda_{\pi}$ as well as the error $E(J,\pi)$ (expressed in percent) of $\upms(J, \pi)$ compared to the exact value of the maximum makespan (which is determined by considering the schedules derived from every job priority assignment). Finally, the errors $E(J,\pi)$ are displayed relative to the corresponding $\lambda_{\pi}$ in Figure~\ref{fig:experiment}. 

\begin{figure}[!h]
\centering
\includegraphics[scale=0.091]{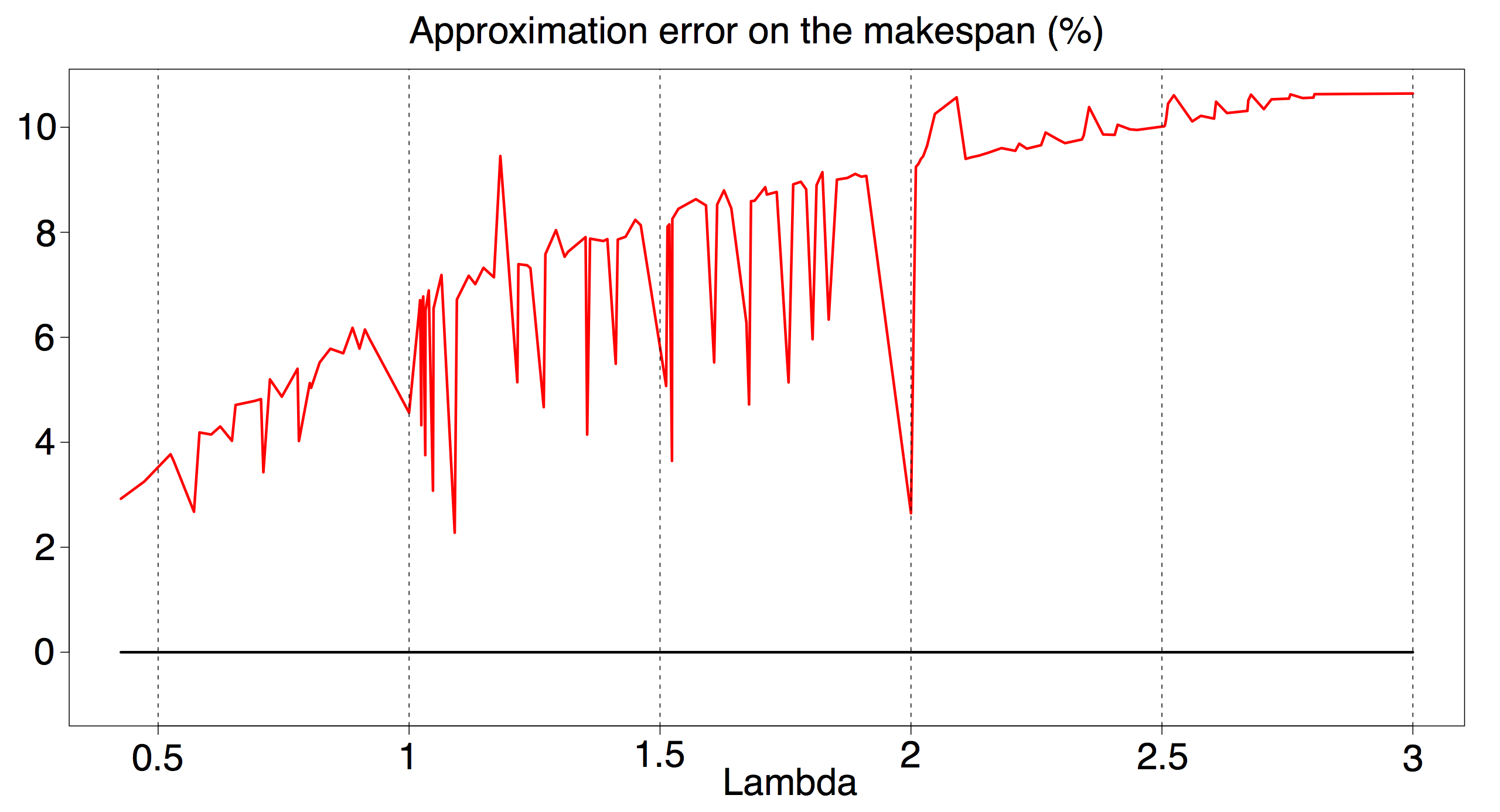}
\caption{Simulation results}
\label{fig:experiment}
\end{figure}

The most important error that we obtained is $10.64\%$ and the minimal one is $2.28\%$. The average error is $7.78\%$ with a squared distance of $2.32\%$. Hence, we believe that this is a promising path to go for more competitive bounds and for practical use.

\section{Conclusion and Future work}\label{Conclusion and future work}

In this paper, the scheduling problem of multi-mode real-time systems upon {\em uniform} multiprocessor platforms is studied. Two protocols for transitioning between every pair of operating modes of the system are specified together with their validity tests. The first proposed protocol $\SUMMSO$ is synchronous in the sense that the tasks of the old- and new-mode are not scheduled simultaneously. The second protocol $\AUMMSO$ is asynchronous in the sense that it allows old- and new-mode tasks to be scheduled together. This study led us to provide the reader with tight bounds for the length of the transient phases during mode transitions. Future work will focus on handling mode-independent tasks, i.e., tasks whose behavior is not affected by the mode changes. \\
\vspace{-0.3cm}

\paragraph{Acknowledgment.} The authors would like to thank Bernard Fortz for taking part in interesting discussions.

\bibliographystyle{acm}
\bibliography{biblio}

\begin{thebibliography}{10}

\bibitem{Andersson0}
{\sc Andersson, B.}
\newblock Uniprocessor {EDF} scheduling with mode change.
\newblock {\em In Proc. of the 12th International Conference on Principles of
  Distributed Systems\/} (2008), 572--577.

\bibitem{BaJo08_2}
{\sc Baruah, S., and Goossens, J.}
\newblock Deadline monotonic scheduling on uniform multiprocessors.
\newblock {\em Proceedings of the 12th International Conference on Principles
  of Distributed Systems\/} (2008), 89--104.

\bibitem{BaJo08_1}
{\sc Baruah, S., and Goossens, J.}
\newblock The {EDF} scheduling of sporadic task systems on uniform
  multiprocessors.
\newblock {\em In Proc. of the Real-Time Systems Symposium\/} (2008), 367--374.

\bibitem{Baruah0}
{\sc Baruah, S., Mok, A., and Rosier, L.}
\newblock Preemptively scheduling hard real-time sporadic tasks on one
  processor.
\newblock {\em In proc. of the 11th IEEE Real-Time Systems Symposium\/} (1990),
  182--190.

\bibitem{CucuJo4}
{\sc Cucu-Grosjean, L., and Goossens, J.}
\newblock Predictability of fixed-job priority schedulers on heterogeneous
  multiprocessor real-time systems.
\newblock {\em Information Processing Letters 110\/} (2010), 399--402.

\bibitem{SGB2001}
{\sc Funk, S., Goossens, J., and Baruah, S.}
\newblock On-line scheduling on uniform multiprocessors.
\newblock {\em In Proc. of the 22nd IEEE Real-Time Systems Symposium\/} (2001),
  183--192.

\bibitem{GLLR1}
{\sc Graham, R., Lawler, E., Lenstra, J., and Rinnooy~Kan, A.}
\newblock Optimization and approximation in deterministic sequencing and
  scheduling: A survey.
\newblock {\em Ann. of Discrete Math., 5:287-326\/} (1979).

\bibitem{Lopez1}
{\sc Lopez, J.~M., Diaz, J.~L., and Garcia, D.~F.}
\newblock Utilization bounds for {EDF} scheduling on real-time multiprocessor
  systems.
\newblock {\em In Journal of Real-Time Systems., vol. 28\/} (2004), 39--68.

\bibitem{Nelis1}
{\sc Nelis, V., Goossens, J., and Andersson, B.}
\newblock Two protocols for scheduling multi-mode real-time systems upon
  identical multiprocessor platforms.
\newblock {\em In Proc. of the 21st Euromicro Conference on Real-Time
  Systems\/} (2009), 151--160.

\bibitem{Pedro}
{\sc Pedro, P.}
\newblock Schedulability of mode changes in flexible real-time distributed
  systems.
\newblock {\em PhD thesis, University of York, Department of Computer
  Science\/} (1999).

\bibitem{Real}
{\sc Real, J., and Crespo, A.}
\newblock Mode change protocols for realtime systems: A survey and a new
  proposal.
\newblock {\em Real-Time Systems Journal, 26(2):161-197\/} (March 2004).

\end{thebibliography}
\end{document}